\documentclass[english,a4paper,11pt]{article}

\usepackage{amsmath, amsxtra, amsfonts, amssymb, amstext, amsthm}
\usepackage{booktabs}
\usepackage{graphics,color}
\usepackage{babel}
\usepackage[noadjust]{cite}
\usepackage{url}
\usepackage[algo2e,ruled,vlined,linesnumbered]{algorithm2e}
\usepackage{todonotes}
\usepackage{fullpage}
\usepackage{hyperref}
\usepackage{paralist}
\usepackage{multirow}

\newtheorem{theorem}{Theorem}[section]
\newtheorem{lemma}[theorem]{Lemma}
\newtheorem{observation}[theorem]{Observation}

\newtheorem{definition}[theorem]{Definition}

\newcommand{\ignore}[1]{}

\newcommand{\RR}{{\mathbb R}}

\newcommand{\NN}{V}

\newcommand{\TT}{\cal T}
\newcommand{\TGP}{T_{\GP}}
\newcommand{\TWU}{T_{\WU}}
\newcommand{\EXGP}{{\cal E}_{\GP}}
\newcommand{\EXWU}{{\cal E}_{\WU}}

\DeclareMathOperator{\HGT}{height}
\DeclareMathOperator{\TIME}{time}
\newcommand{\HGP}{h_{GP}}
\newcommand{\HWU}{h_{WU}}

\newcommand{\vecb}{\vec{b}\,}
\newcommand{\vecc}{\vec{c}\,}
\newcommand{\vecd}{\vec{d}\,}
\newcommand{\vece}{\vec{e}\,}
\newcommand{\LL}{\textsf L}
\DeclareMathOperator{\REQ}{Delay}
\DeclareMathOperator{\EVEN}{\textbf{even}}
\DeclareMathOperator{\ODD}{\textbf{odd}}

\DeclareMathOperator{\GP}{GP}
\DeclareMathOperator{\WU}{WU}

\newcommand{\CEIL}[1]{\lceil {#1}\rceil}
\newcommand{\FLOOR}[1]{\lfloor {#1}\rfloor}
\DeclareMathOperator{\uniform}{uniform}
\DeclareMathOperator{\sender}{sender}
\DeclareMathOperator{\target}{target}
\DeclareMathOperator{\DIS}{DIS}
\DeclareMathOperator{\Sec}{Sec.}

\DeclareMathOperator{\PP}{Path}

\DeclareMathOperator{\E}{E}

\newcommand{\N}{\mathbb{N}}
\newcommand{\R}{\mathbb{R}}

\begin{document}

\title{Simple and Optimal Randomized Fault-Tolerant Rumor Spreading}

\author{Benjamin Doerr$^1$,
Carola Doerr$^{2,3}$
Shay Moran$^{4,5}$,
Shlomo Moran$^5$}
\date{
$^1$LIX, \'Ecole Polytechnique, Palaiseau, France\\
$^2$Sorbonne Universit\'es, UPMC Univ Paris 06, UMR 7606, LIP6, F-75005, Paris, France\\
$^3$CNRS, UMR 7606, LIP6, F-75005, Paris, France\\
$^4$Max Planck Institute for Informatics, Saarbr\"ucken, Germany\\
$^5$Computer Science Dept., Technion - Israel Institute of
Technology, Haifa, 32000 Israel\\[2ex]
\today
}

\maketitle

\begin{abstract}
We revisit the classic problem of spreading a piece of information in a group of $n$ fully connected processors. By suitably adding a small dose of randomness to the protocol of Gasienic and Pelc (1996), we derive for the first time protocols that
(i) use a linear number of messages,
(ii) are correct even when an arbitrary number of adversarially chosen processors does not participate in the process, and
(iii)  with high probability have the asymptotically optimal runtime of $O(\log n)$ when at least an arbitrarily small constant fraction of the processors are working.
In addition, our protocols do not require that the system is synchronized nor that all processors are simultaneously woken up at time zero, they are fully based on push-operations, and they do not need an a priori estimate on the number of failed nodes.

Our protocols thus overcome the typical disadvantages of the two known approaches,
algorithms based on random gossip (typically needing a large number of messages due to their unorganized nature) and
algorithms based on fair workload splitting (which are either not  {time-efficient} or require intricate preprocessing steps plus synchronization).
\end{abstract}

\sloppy{
\section{Introduction}
\label{sec:intro}

 Disseminating information to all nodes of a network is one of the basic communication primitives. Basically all collaborative actions in networks imply that some information has to be sent to all nodes, and surprisingly complex tasks like computing aggregates can be reduced to essentially solving a dissemination problem~\cite{MAS08}.
We are interested in disseminating a single piece of information, the \emph{rumor}, to all $n$ nodes in a communication network in which all nodes can exchange information with each other but where individual nodes can initially crash; i.e., they do not participate in the rumor spreading process. More precisely, we study dissemination protocols that are robust against \emph{adversarial initial node crashes}~\cite{TKM89}.

\subsection{Previous Results}

 Rumor spreading protocols that are robust against adversarial node failures have been studied mainly in complete communication networks.
In such networks, essentially two types of fast fault-tolerant rumor spreading protocols have been proposed:
(i) \emph{whispering protocols,} which assume that in every round every informed node passes the information to at most one other nodes,
and (ii) \emph{gossip-based protocols,} which build on the paradigm that nodes call randomly chosen others. The latter, due to their randomized nature, usually are highly robust against all kinds of faults, typically at the price of a higher communication effort and non-trivial termination criteria.

\paragraph{Whispering Protocols:}
It is easy to see that in fault-free communication networks there are protocols disseminating a rumor in $\lceil \log n \rceil$ communication rounds using a total of $n-1$ messages and that both these measures are strictly optimal. A simple protocol for $n = 2^k$ nodes indexed by the numbers from $0$ to $n-1$ would be that in round~$i$, each node~$x$ having the rumor calls node $x$~XOR~$2^{i-1}$ and forwards the message to it. From the sender ID the recipient of a message can infer the round number~$i$, and can thus decide when to stop forwarding the message. Hence this protocol indeed uses only $n-1$ messages in total. This algorithm, like the other ones mentioned in this subsection but in contrast to gossip-based algorithms discussed further below, maintains the {\em whispering property}~\cite{GasPelc96,DiksPelc00}: in each round, the edges along which the rumor is transferred form a matching. It
has two further advantages, namely
(i) it requires no synchronization, in the sense that its correctness does not assume the existence of a global clock, or any restriction on the relative speeds of the processors, and
(ii) nodes know when to stop forwarding the rumor.

The downside of this simple approach is that it is not at all robust. If a node is not available (``crashed''), then all other nodes that would have been informed via it will remain uninformed.
This problem was overcome in the preliminary version of~\cite{GasPelc96},\footnote{The simple workload splitting protocol of Gasieniec and Pelc which we describe and use below appeared only in the preliminary version of~\cite{GasPelc96} and not in the journal paper itself. As indicated in the reference section, this preliminary version is available online at \url{http://citeseerx.ist.psu.edu/viewdoc/summary?doi=10.1.1.49.3838}.} where a protocol is presented that is strictly optimal if no failures occur and which informs all nodes in the network even if an arbitrary number $f$ of nodes do not participate in the collaborative process.
The number of messages sent in all cases is $n-1$. In this result, as in most other fault-tolerant rumor spreading algorithms (but unlike in standard randomized rumor spreading), it is assumed that a node calling a crashed node learns that his call was unsuccessful.

The stability of the algorithm from~\cite{GasPelc96} comes at the cost of an increased runtime: if $f$ nodes are crashed, the protocol may terminate only after $f + \lceil \log (n-f) \rceil$ rounds. That is, the increase of the runtime bound is in the worst case almost equal to the number of crashed nodes. This is certainly unsatisfactory for large $f$, e.g., $f \in \omega(\log n)$.

Subsequent fault-tolerant whispering protocols reduced this running time by adding an {\em opening phase}\footnote{We distinguish between {\em opening phase}, which is repeated each time a rumor is spread, and a {\em preprocessing phase}, which is performed only once, when the network is established.} which precedes the actual spreading of the rumor, and connects a large portion of the non-faulty processors in an appropriate subnetwork.
 The added opening phase comes with a price: (i) it is tailored for a linear lower bound $\alpha n$ on the number of non-faulty processors,
  which needs to be determined by the user in advance; hence, when the number of non-faulty processors is smaller than $\alpha n$, these protocols may run into a deadlock and the rumor is not guaranteed to reach all non-faulty processors; also, the actual running time and number of messages of these protocols are determined by the $\alpha n$ bound rather than by the actual number of faulty processors;
  (ii) unlike the fault-free whispering protocols, it assumes the presence of a global clock and requires synchronization; (iii) it requires that all the non-faulty nodes are simultaneously activated at time 0.

   \cite{GasPelc96} introduces such an opening phase which runs in
   $O(\log ^2 n)$ time.
   A more intricate opening phase was later introduced in~\cite{DiksPelc00};  to the best of our knowledge,~\cite{DiksPelc00} is the only published paper which shows that for any fixed constant $\varepsilon$ there is a protocol which can tolerate up to $\varepsilon n$ node failures, and whose time and message complexities are both asymptotically optimal.
   The algorithm of~\cite{DiksPelc00} requires that a certain, rather complex, virtual expander is stored by the network nodes during the system setup in a preprocessing phase. This expander is used in the opening phase for constructing the paths along which the rumor is communicated. The construction of the expander is based on the explicit expanders of~\cite{RamGraphs88}, and on the properties of these expanders presented in~\cite{Upfal94}.

\paragraph{Gossip-based Protocols:}
In gossip-based communication protocols the nodes of the network are assumed to have access to random bits, which allows them to decide at random which other node to contact. Gossip-based communication algorithms are therefore also called \emph{randomized rumor spreading} protocols. Randomized rumor spreading has been analyzed in various variants for different network topologies. Despite the very simple approach of talking to random neighbors, these protocols often achieve a surprisingly good runtime combined with extreme robustness.
Their main advantages over the fault-tolerant whispering protocol of~\cite{DiksPelc00} is that they avoid the need for an opening phase and for storing intricate subnetworks in a preprocessing time, they tolerate runtime failures (i.e., it may happen that a node works initially but stops working after some time), they do not assume that failures are detected, and usually they also do not assume that processor names are mutually known. On the negative side, they do require asymptotically larger message complexities, and they typically lack a simple termination criterion (i.e., the nodes do not know when every node is guaranteed to have learned the rumor so that they can stop spreading the message).
In this section, we briefly describe the results that are relevant for our work on robustness against adversarial failures in complete graphs.

The first rumor spreading result is due to Frieze and Grimmett~\cite{FG85}, who studied the simple protocol in which each informed node calls in every round a neighbor that is chosen uniformly at random from the list of all its neighbors (\emph{synchronized push-protocol}).
 Pittel~\cite{Pi87} showed that the round complexity of this protocol is $\log n  + \ln n  + h(n)$, where $h(n)$ is any function tending to infinity. Note that randomized rumor spreading in this push-model violates the whispering property, but when counting only messages which carry the rumor (see Section~\ref{subsec:rsp}), this violation can be undone by assuming that nodes accept only one incoming call.

The first to analyze rumor spreading as communication protocol were Demers et al.~\cite{DGH+87}. They studied the problem of maintaining the consistency of replicated databases. In applications like this, where one may assume that updates are to be disseminated frequently, also a push-pull randomized rumor spreading protocol makes sense. Here all nodes and not only those already knowing the rumor call random neighbors, allowing that uninformed nodes ``pull'' information from informed ones.

A possible weakness of protocols using randomized pull operations is the inherent violation of the whispering property: Many uninformed nodes may randomly select the same informed one, thus forcing the selected node to forward the rumor to many neighboring nodes in a single round.\footnote{Using, e.g.,~\cite{BallsBins98}, it is not hard to see that in a push-pull gossip-based algorithm, in each round in which the fraction of informed processors is bounded away from~0 and from~1, with high probability some informed processor is contacted by $\Omega(\log(n) /\log\log n)$ uninformed neighbors each of which wants to receive the rumor from the informed node. This node therefore has to forward the rumor to $\Omega(\log(n) /\log\log n)$ neighbors in one round.}
In models which allow a processor to send at most one message each round, this may result in a considerable increase in running time.

The robustness of randomized rumor spreading comes at a price:
a slightly higher dissemination time when no failures occur, and more importantly  a relatively large number of messages sent until the rumor is disseminated, and a large number of additional messages caused by the fact that in the basic protocol the nodes do not know when to stop sending out messages: In independent randomized rumor spreading in the push-model, only after $\Theta(n \log n)$ messages are sent, the rumor is known to all vertices.

To summarize, randomized gossip based algorithms have asymptotically larger message complexities than optimal deterministic whispering algorithms, but they maintain stronger robustness, and are considerably simpler.

\subsection{Our Results}
The aim of our work is to design simple fault-tolerant rumor spreading protocols with optimal message and time complexities.
We achieve this by adding a natural randomization to the elegant whispering protocol of~\cite{GasPelc96}.
Our protocols inherit from the algorithm in~\cite{GasPelc96} the following properties.
\begin{itemize} \itemsep-0.2em
    \item They do not need to construct and store, in a preprocessing time, an intricate network structure.
    \item  They do not need an opening phase, or simultaneous wakeup of all processors at time 0.
    \item They use only push operations.
     \item They are asynchronous, in the sense that they do not require that the nodes perform actions at the same time or speed.
     \item They have a very simple termination criterion.
      \item They maintain always strictly optimal message complexity; i.e., they use $n-1$ messages to inform all non-crashed nodes in the network.
    \item They are correct for any number of non-faulty processors.
       \item Their  runtimes  are determined by the actual fraction $p$ of the non-faulty processors, and not by a predetermined lower bound on $p$.
\end{itemize}

In particular, our protocols are the first rumor spreading
protocols combining the following three properties:
(i) they disseminate the rumor to all nodes in the network, regardless of how many nodes are crashed by an adversary,
(ii) if an arbitrarily small constant fraction $p$ of nodes (including the starting node) is working correctly, the rumor is spread in logarithmic time, and
(iii) the number of messages sent is linear in $n$ (the protocols actually use strictly minimal $n-1$ messages).

Table~\ref{tab:comparison} compares the properties and advantages of the two whispering protocols from~\cite{DiksPelc00} and~\cite{GasPelc96}, respectively, with those of our new whispering protocols. In this table we state only asymptotic time complexities as there is a trade-off between the exact running time and the probability to achieve this runtime bound. These trade-offs are detailed in the respective theorems. The bit complexity counts the total and maximum number of bits that need to be appended to the original message, respectively.

\begin{table}
\begin{center}
\begin{tabular}{llllllll}
& & \cite{DiksPelc00} & GP~\cite{GasPelc96} &
$\GP_{\uniform}$ &
$\GP_{\Sec~\ref{sec:derandomize}}$ &
$\GP^a_{\Sec~\ref{sec:kwiseind}}$ &
$\GP^b_{\Sec~\ref{sec:kwiseind}}$\\
\hline
\multicolumn{2}{l}{time complexity}
& $O(\log n)$ & {\small $O( n)$} & $O(\log n)$ & $O(\log n)$& $O(\log n)$ & $O(\log n)$\\
\hline
\multicolumn{2}{l}{ message complexity}
& $O(n)$ & $n-1$ & $n-1$ & $n-1$ & $n-1$ & $n-1$ \\
\hline
\multirow{2}{*}{  bit complexity}
& {\small total}
& {\small $O(n \log n)$ }&{\small $O(n \log n)$} & {\small $O(n \log^2 n)$} & {\small $O(n \log n)$ }& {\small $O(n \log^{2}n )$ }& {\small $O(  n \log^{2}(n)  h(n))$} \\
& max
& {\small $O(\log n)$ }& {\small $O(\log n)$} & $O(n)$ & {\small $O(\log n)$} & {\small $O(\log^{2}n )$ }& {\small $O(  \log^{2}(n)  h(n))$} \\
\hline
\multicolumn{2}{l}{preprocessing}
& yes & no & no & yes & no & no\\
\hline
\multicolumn{2}{l}{extra memory}
& yes & no & no & yes & no & no\\
\hline
\multicolumn{2}{l}{opening phases}
& yes & no & no & no & no & no \\
\hline
\multicolumn{2}{l}{  simultaneous wakeup}
& yes & no & no & no & no & no \\
\hline
\multicolumn{2}{l}{global clock}
& yes & no & no & no & no & no\\
\hline
\multicolumn{2}{l}{ dependance on $p$}
& yes & no & no & no & yes & no\\
\hline
\end{tabular}
\end{center}
\caption{Comparison of the fault-tolerant whispering protocols.
The asymptotic time complexities are stated for the case that at least $pn$ nodes are non-faulty, for some constant $p\in(0,1)$, and they hold in the worst case for the deterministic protocols and w.h.p. for the randomized ones.
``dependance on $p$'' indicates that the protocol is determined by the value of $p$; for the two protocols which depend on $p$,  the one in~\cite{DiksPelc00} is guaranteed to be correct only if the number of non-faulty nodes is at least $pn$, and the one in $\GP^a_{\Sec~\ref{sec:kwiseind}}$ is guaranteed to have $O(\log n)$ time complexity only in this case.
In the rightmost column, $h$ can be any slow growing function which tends to infinity.}
\label{tab:comparison}
\end{table}

\paragraph{Analysis of Random Initial Node Failures (Section~\ref{sec:wakeupandcoupling}):}
We first show that for random node crashes, the basic protocol of~\cite{GasPelc96}, denoted $\GP$, has a much better performance than what the worst-case bound in~\cite{GasPelc96} states. In particular, when each node is crashed with constant probability $0<c<1$ independently at random, then with high probability the algorithm terminates within $\Theta(\log n)$ rounds.
We prove this result by first introducing an intermediate failure model, the wakeup model.
We believe that the wakeup model itself is of independent interest.

\paragraph{Simple Randomized Version of the GP Algorithm (Section~\ref{sec:fnm}):}
For adversarial node failures, the robustness of the GP algorithm against random initial node crashes suggests the following straightforward {\em randomized} solution: The node which originally holds the rumor picks a random permutation $\pi$ of the other nodes and initiates the GP protocol with node labels permuted according to $\pi$.
This gives the same time bounds as for random node failures. The downside is that to make the other nodes adopt this strategy, sufficient information on the permutation $\pi$ has to be communicated to the other nodes as well.
This can be achieved by adding in total $O(n \log^2 n)$ bits to all the messages, with at most $n$ bits appended to a single message.

\paragraph{Randomized Version with Logarithmic Message Sizes (Section~\ref{sec:derandomize}):}
 The message sizes can be reduced to $O(\log n)$, which is the messages size in the original GP protocol and the one in~\cite{DiksPelc00}: We prove that instead of choosing the permutation randomly from all permutations, it suffices to choose the permutation uniformly from a set of only $O(n h(n)/\log n)$ {\em random} permutations, where $h \in \omega(1)$ is an arbitrary function tending to infinity. (The number of permutations can be varied to adjust runtimes and failure probabilities, see Theorem~\ref{thm:derandomized} for the details.) This allows to encode the permutation via only $\Theta(\log n)$ bits.
This approach can be implemented by choosing, for an arbitrary function $h \in \omega(1)$, $O(n h(n)/\log n)$ random permutations, communicating them to all nodes in the network, and storing them at all processors. This preprocessing procedure is repeated each time processors join or leave the network. Thus, this protocol is particularly appealing when communication is expensive, memory is cheap, and processors are not added or removed from the network too often.

\paragraph{Preprocessing-Free Version (Section~\ref{sec:kwiseind}):}
While the algorithm from Section~\ref{sec:derandomize} is certainly an efficient way to disseminate rumors in a network whose topology does not change too often and in which messages need to be exchanged frequently, there are many situations in which the required preprocessing phase is undesirable.
 In Section~\ref{sec:kwiseind}, we avoid the preprocessing stage by a natural application of \emph{adaptive $k$-wise almost independent permutations}~\cite{KaplanNR09}.
In a nutshell, this is a probability distribution over a small number of permutations with two properties:
 (i) a permutation chosen randomly from this set has, in the context of our application, almost the same features as a uniform random permutation, and
 (ii) each of the permutations can be computed efficiently from an encoding of $O(k \log n)$ bits.
Consequently, we do not need to communicate or store the permutations, but it suffices to append to each rumor the $O(k \log n)$ bits encoding the  permutation selected by the node initiating the rumor.
The value of $k$ in this algorithm is $O(\log n)$, where the constant of proportionality is determined by the fraction $p$ of non-faulty processors.  We present also a variant of this algorithm in which the value of $k$ does not depend on $p$, in the price of slightly increased asymptotic messages size. These variants are denoted in Table~\ref{tab:comparison} by $\GP^a_{\Sec~\ref{sec:kwiseind}}$ and $\GP^b_{\Sec~\ref{sec:kwiseind}}$, respectively.

\section{Preliminaries}
\label{sec:preliminaries}

Before we present a few basics about rumor spreading protocols, let us briefly fix the notation used throughout this work.
We consider executions of rumor spreading algorithms by $n$ processors whose IDs $(0,1,\ldots,n-1)$ are mutually known, where 0 is the start processor, and where each node can contact each other node.

We use the following notation: For a sequence $s=(s_1,s_2,\ldots )$, $\ODD(s)= (s_1,s_3,\ldots )$ is the subsequence of the odd indexed elements of $s$, and $\EVEN(s)=(s_2,s_4,\ldots)$ is the subsequence of the even-indexed elements of $s$.
 For a binary vector $\vecb$, $|\vecb|_0$ is the number of zeros in $\vecb$, and $|\vecb|_1$ denotes the number of ones in $\vecb$.

 For a rooted tree $T$, 
 $\HGT(T)$ is the height of $T$, i.e., the maximum length of a path from the root to a leaf.

 For $n \in \N$ ($\N$ denotes the positive integers)
 we abbreviate $[n]:=\{1,2,\ldots, n\}$.
 By $S_n$ we denote the set of all permutations of the set $[n]$.

 By $\ln$ we denote the natural logarithm to base $e$. 
 All other logarithms are to base 2.

An event is said to happen with high probability if it happens with probability $1-o(1)$.

 \subsection{Rumor Spreading Protocols}\label{subsec:rsp}

We give a short description of rumor spreading protocols, which aims at formalizing the concepts used in this paper.

Let the undirected graph $G = (V, E)$ describe the underlying communication network, that is, nodes of this graph represent processors and a direct communication between two processors is possible if and only if there is an edge between the corresponding nodes. Let $n := |V|$.

A {\em synchronous} execution of a rumor-spreading algorithm on $G$ consists of rounds $\RR_1,\RR_2,\ldots$. A round $\RR_t$ is initiated by a set of processors $\NN_t\subseteq V$ (the exact nature of $\NN_t$ depends on the model assumed and/or on the specific algorithm): 
  each processor $u\in \NN_t$  sends a  $(u,v)$ {\em communication request} (in short ``$(u,v)$ request'') to one of its neighbors $v$; the request contains a bit informing $v$ whether $u$ holds the rumor already. A $(u,v)$ request is {\em valid} if exactly one of $u$ and $v$ holds the rumor. Each communication  request is received at the same round in which it was sent.
   After receiving all the requests sent to it at $\RR_t$, each processor~$v$ may (but does not have to) {\em approve} some of the valid requests that it has received.  The round $\RR_t$ is then
 completed by transferring the rumor along the edges of the approved requests.\footnote{Some models assume that an informed processor $u$ always sends the rumor
  on the selected edge $(u,v)$, even if $v$ is already informed.}
  The execution \emph{terminates} at time $t$ if $\NN_t\ne\emptyset$ and $\NN_{s}=\emptyset$ for all $s>t$.
 We call $t$ the \emph{time} (or \emph{round}) \emph{complexity} of the execution of the rumor spreading algorithm. Note that in some other works, in particular those on gossip based randomized rumor spreading, only the first time at which all processors know the rumor is regarded---this event can happen much before all nodes stop performing any action.
  In rumor spreading algorithms which assume that all executions are synchronous, the round number can be used by the local programs of the processors  for checking termination etc. In such algorithms the round number  can thus be viewed as a discrete-time global clock.
   In contrast, {\em asynchronous} executions are executions which do not have such a global clock. For the sake of this paper,  asynchronous executions of a rumor spreading protocol  differ from synchronous executions by the following:
   (a) the round numbers are not known to the processors, and
   (b) the processors can have different speeds, which may vary in each round, i.e., a $(u,v)$ communication request sent at round $\RR_t$ can be received by $v$ at round $\RR_{t+s}$, where $s$ can be any nonnegative integer which depends on $u$, $v$, and $t$.\footnote{For a formal definition of asynchronous protocols see, e.g.,~\cite[Chapter 14]{Lynch96}.} Note that if processors use round numbers in their programs, the resulted algorithm may err in asynchronous executions. A rumor spreading protocol is said to be \emph{asynchronous} if it is correct also in asynchronous executions.

Let $E_t$ denote the set of edges along which requests are sent in $\RR_t$, and let $F_t$ denote the set of edges of the approved requests, along which the rumor is transferred at $\RR_t$ (thus $|E_t|\leq|\NN_t|\leq n$ and  $F_t\subseteq E_t$).
A rumor spreading algorithm satisfies the {\em whispering property} if $F_t$ always forms a matching, meaning that each processor may either send or receive at most one copy of the rumor at each round.  Some authors actually require any rumor spreading algorithm to satisfy the whispering property, see, e.g.,~\cite{GasPelc96,DiksPelc00}.

Besides the time complexity, the communication effort and the robustness against faults are two further important performance measures.
 There are some variants of the definition of \emph{message complexity} of rumor spreading algorithms.
The strictest definition  counts all communication requests, i.e. $\sum_t{|E_t|}$, e.g.,~\cite{GasPelc96,DoerrFouzReducing}. A more permissive definition  assumes that communication between uninformed processors is given for free due to frequent injections of other rumors~\cite{KSSV00,CHKM12},
and hence it reduces to
$\sum_t |\{(u,v): (u,v)\in E_t \mbox{~and either } u \mbox{ or } v
 \mbox{ holds the rumor }\}|.$
As will be noted soon, our algorithms have the minimum possible message complexity by both definitions.

 The \emph{faults} assumed in this paper are  {\em initial crash failures}: A processor is faulty in a given execution if it never sends a message during the execution. We consider two types of failure policies, associated with a success parameter $p\in(0,1)$: {\em random failures}, in which each process may fail independently with probability $1-p$, and {\em adversarial failures}, in which the adversary may fail (before the execution of the algorithm starts) any subset of up to $(1-p)n$ processors, excluding the start processor. An $(i,j)$ request is {\em failed}  if $j$ is faulty, and it is {\em successful} otherwise. Note that in our synchronized model, a faulty node $j$ is identified by not responding to an $(i,j)$ request.

\subsection{The Algorithm of Gasieniec and Pelc}
\label{subsec:GP}
We use the following variant of the divide-and-conquer
  algorithm of Gasieniec and Pelc~\cite{GasPelc96}, to be denoted $\GP$. Initially  the start processor 0 holds a list $(1,2,\ldots,n-1)$ of all uninformed processors, 
and all other processors hold empty lists.  At each round, each processor $i$ which holds a nonempty list $(j_1,\ldots,j_k)$, sends an $(i,j_1)$ request and deletes $j_1$ from its list.
If the request is successful then $i$ also sends to $j_1$ the rumor, appends to it the list $\EVEN(j_2,\ldots,j_k) = (j_3, j_5, \ldots)$, and sets its own list to  $\ODD(j_2,\ldots,j_k)$.
Thus, in this case, the next round starts with $i$ holding the list $\ODD(j_2,\ldots,j_k)$ and
processor $j_1$ holding the list $\EVEN(j_2,\ldots,j_k)$.
The algorithm terminates when all processors hold empty lists.\\
 {\bf Implementation note:} Observe that each list of the form $\EVEN(j_2,\ldots,j_k)$ generated during the algorithm is an arithmetic progression whose difference is $2^m$ for some integer $m\le \log n$.  Sending such a list can be done by sending the first element $j_3$ and the exponent $m$.
 Overall, this requires an addition of less than $2\log n$ bits to the rumor.

Note that this protocol automatically ensures that (i) each node receives at most one communication request per round (hence the whispering property is satisfied),
 (ii) only requests from informed nodes to uninformed ones are issued (hence there is no reason not to approve a request), and (iii) the protocol terminates as soon as all processors know the rumor.

 The optimality of the message complexity of the $\GP$ algorithm (under the different variants of ``message complexity'' discussed in Section~\ref{subsec:rsp}) is implied by the following straightforward observation.

\begin{lemma}[\cite{GasPelc96}]
\label{lem:optimalmessage}
The $\GP$ algorithm performs the minimum possible number 
 of communication requests, namely $n-1$ communication requests in each possible execution.
\end{lemma}

In the presence of $f$ crashed nodes, the time complexity of the $\GP$ algorithm is given by the following lemma.

\begin{lemma}[\cite{GasPelc96}]
\label{lem:GPfailure}
For up to $f$ initial node failures the time complexity of the $\GP$ algorithm is at most $f + \lceil\log(n-f)\rceil$.
This bound is tight if processors $1, \ldots, f$ are failed.
\end{lemma}

\subsection{Reminder: Chernoff's Bounds}%
\label{subsec:Chernoff}

We apply several versions of Chernoff's bound, see, for example,~\cite{MitzenmacherU05} or~\cite{DubhashiP98} for an introduction to these elementary tail bounds.

\begin{theorem}[Chernoff's bounds]
\label{thm:Chernoff}
Let $X=\sum_{i=1}^n{X_i}$ be the sum of $n$ independently distributed random variables $X_i$, where each variable $X_i$ takes values in $[0,1]$.
Then the following statements hold.
\begin{align}
& \forall t>0:\label{cher:i}\\
	  &\quad \Pr[X > \E[X]+t]\leq \exp(-2t^2/n), \nonumber\\
		&\quad \Pr[X < \E[X]-t]\leq \exp(-2t^2/n)\,.\nonumber\\
& \forall 1> \varepsilon >0:\label{cher:ii}\\
		&\quad \Pr\big[X< (1-\varepsilon) \E[X] \big] \leq \exp\big(- \varepsilon^2 \E[X]/2 \big), \nonumber\\
	 	&\quad \Pr\big[X> (1+\varepsilon) \E[X] \big] \leq \exp\big(- \varepsilon^2 \E[X]/3 \big)\,. \nonumber\\
& \forall t>2 e \E[X]:
	  \Pr[X>t] \leq 2^{-t}\,. \label{cher:iii}
\end{align}
\end{theorem}

Chernoff's bounds apply also to random geometric variables. A proof of the following theorem can be found, e.g., in~\cite[Theorem~1.14]{Doerr11bookchapter}.

\begin{theorem}[Chernoff's bound for random geometric variables]
\label{thm:geometricconcentration}
Let $p \in (0,1)$. Let $X_1,\ldots, X_n$ be independent geometric
random variables with
$\Pr[X_i = k] = (1-p)^{k-1}p$ for all
$k \in \N$.
Let
$X :=\sum_{i=1}^n{X_i}$.

Then for all $\delta>0$,
\begin{align*}
\Pr[X \geq(1+\delta) \E[X]]
\leq
\exp\left(- \frac{\delta^2 (n-1)}{2 (1+\delta)}\right)\,.
\end{align*}
\end{theorem}

\section{Random Failure Analysis of the \texorpdfstring{$\GP$}{GP} Algorithm via a New Random Wakeup Model} \label{sec:wakeupandcoupling}

In this section we show that the $\GP$ algorithm has a much better performance against {\em random} node failures than the worst case performance given in  Lemma~\ref{lem:GPfailure} against {\em adversarial} node failures.
We assume that each processor may fail with probability $1-p$ independently.
It is not hard to see that the expected runtime is bounded from below by the solution to the recursive formula $F(1)=0; F(n) = p\cdot F(n/2)+(1-p)\cdot F( n-1)+1$, which is $\log n/p +O(1)$.
On the other hand, we show that every processor is informed after $3.5 \log n/p$ rounds, with high probability.

We refer to $p$ as the \emph{success rate,} and to $1-p$ as the \emph{failure rate}.
Due to the sequential nature of the $\GP$ protocol, even a very small change in the failure pattern (that is, the set of failed nodes) may imply a large change in the time complexity.
This makes a straightforward analysis of this model a bit tricky.
To ease the analysis, we start by considering a similar protocol in a simpler model, the {\em random wakeup model}, which we believe to be of independent interest.
We then transfer the results to the standard random node failure model by coupling the models.

\subsection{The Random Wakeup Model}
\label{sub:wakeup}
We regard the following divide-and-conquer wake-up protocol, which
is inspired by the GP algorithm. The start processor 0 starts with the
list $(1, 2, \ldots, n-1)$ of nodes to be informed. It sends in every round a communication (``wakeup'') requests to processor $1$, until this processor is woken up. It then forwards to it the rumor, appended by the list $\EVEN(2,..., n-1)$, thus keeping for itself the
list $\ODD(2,...,n-1)$ as its todo-list. It then tries to wake up processor $2$ in the next round, and so on.
In this model, each wakeup request is successful with probability~$p$, independently of previous requests.
Hence in the implied rumor spreading algorithm, to be denoted $\WU$, whenever $u$ selects an edge $(u,v)$, it repeatedly sends $(u,v)$ requests until $v$ is woken up.
Informally, the time complexity of the algorithm in this model is larger than in the standard initial-failures model, since in the standard model only one request is sent to each processor. A formal proof of this statement is given in Section~\ref{subsec:monotonicity}. Note also that similar to the $\GP$ algorithm, the $\WU$ algorithm performs the minimum possible number of communication requests in each execution: $n+f-1$ requests when there are $f$ failed wakeup messages.

The time complexity of the random wakeup model is easier to analyze since the implied $\WU$ algorithm sends communication requests along a fixed set of edges, which is independent of the specific failure pattern.
%
 For analyzing this time complexity
  we represent the $\WU$ algorithm by a full binary tree $\TT$ with $n$ leaves, in which each vertex $x$ is labeled by a processor name $\LL(x)\in\{0, \ldots, n-1\}$ according to the following scheme (cf. Figure~\ref{wake_up}).
The leaves of $\TT$ are labeled by the processor names $0,\ldots,n-1$, according to some arbitrary but fixed order.
 The labeling of an internal vertex $x$ with children $y,z$ is
  $\LL(x)=\min\{\LL(y),\LL(z)\}$.  Thus $\LL(r)=0$ (where $r$ is the root of the tree), and for each processor $k$, the vertices of $\TT$ labeled by $k$ form a directed path, $\PP_k$, ending at a leaf of $\TT$.

The algorithm for processor $k\in\{0,\ldots,n-1\}$ implied by the above labeled tree $\TT$ is the following: After receiving the rumor, $k$ moves along the vertices of $\PP_k$. When $k$ steps on a non-leaf vertex $x\in\PP_k$ with children $y,z$,  it repeatedly sends communication requests to $j=\max\{ \LL(y),\LL(z)\}$ until $j$ wakes up.

\begin{figure}[t]
\begin{center}
  \includegraphics[width=0.7\columnwidth]{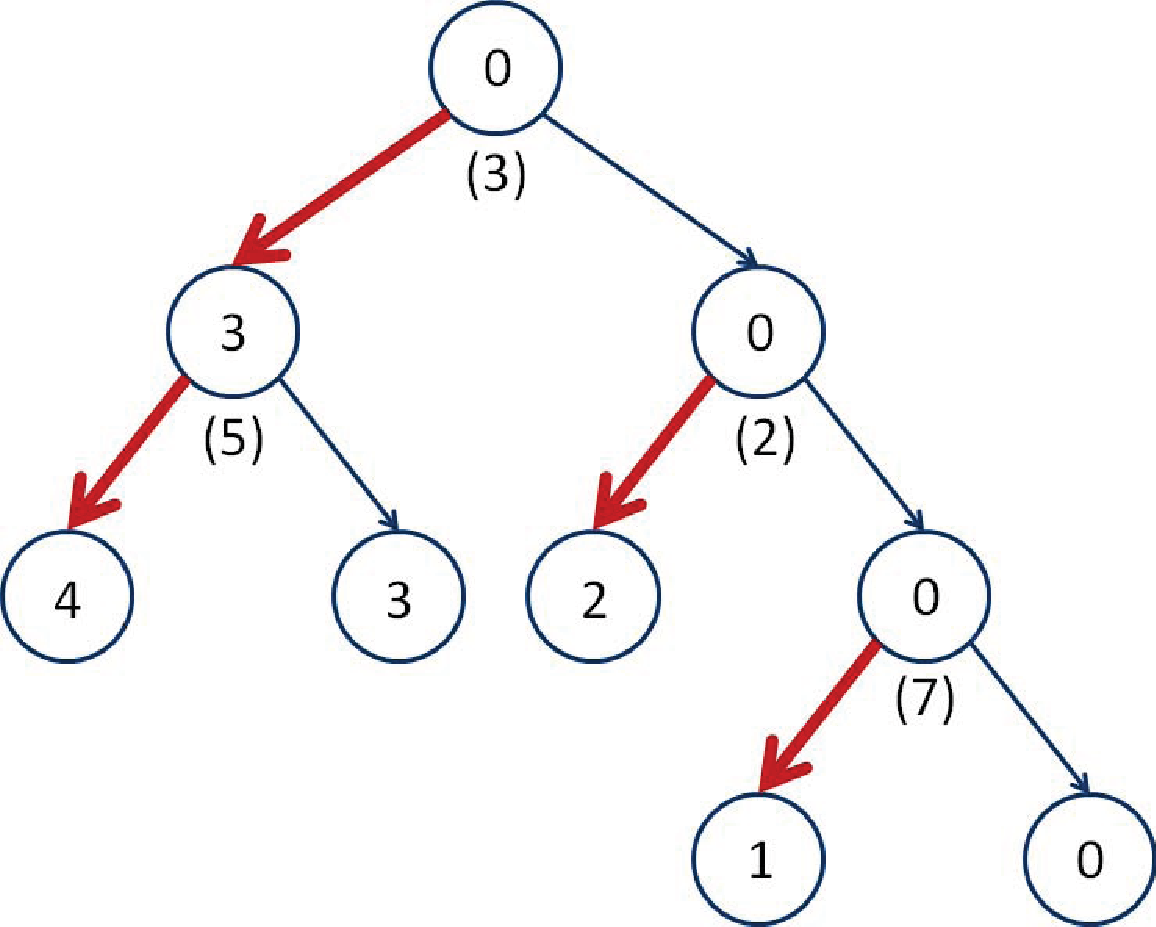}
 \caption{\textbf{Illustration of the rumor spreading in the random wakeup model for five processors:}
Each vertex of $\TT$ is labeled by a processors name.
The red bold edges indicate rumor transfers.
Thus processor 0 always transfers the rumor to processors 3, 2, and 1 (in this order).
The numbers in parentheses beneath internal vertices indicate the number of  wakeup calls in a specific execution.
That is, in the depicted execution processor 1 woke up only by the seventh $(0,1)$ request. The time complexity of this execution is $3+2+7=12$.}
\label{wake_up}
 \end{center}
\end{figure}

Consider now a specific execution $\EXWU$ of the above random wakeup algorithm. For each internal vertex $x\in\TT$ with children $y,z$, let $k_x=\LL(x)$ and $j_x=\max\{\LL(y),\LL(z)\}$. Denote by $\REQ(x)$ the number of $(k_x,j_x)$ requests sent by $k_x$ in $\EXWU$. Then  $\REQ(x)$ is a geometric random variable with probability $p$, that is
 $\Pr[\REQ(x)=\ell]=(1-p)^{\ell-1}p$ for all positive integers $\ell$, and $\E(\REQ(x))=1/p$.

  For a processor $j\in[0\ldots n-1]$, let $P_j$ be the  path from the root $r$ of $\TT$ to the (unique) leaf labeled by $j$, and let $\REQ(P_j):= \sum_{x\in P_j }\REQ(x)$. Then the time complexity of $\EXWU$ is given by
\[
  \TIME(\EXWU)= 
\max_{j\in[0\ldots n-1]}   \{\REQ(P_j)\}.
\]

\subsection{The Time Complexity of the Random Wakeup Model}
\label{subsec:fnp}

\begin{theorem}
\label{thm:wakeuprandomfailure}
Let $c>1$ be a constant and let $p \in (0,1)$ be arbitrary (possibly $p=1-o(1)$).

With probability at least\\
$1-n\exp\left(-\frac{(c-1)^2}{2c} (\lceil \log(n-1) \rceil -1) \right)$,
the $\WU$ algorithm with success rate $p$ has delivered the rumor to all processors after $\frac{c}{p} (\lceil \log(n-1) \rceil+1)$ rounds.
\end{theorem}

The success probability in
Theorem~\ref{thm:wakeuprandomfailure} becomes $1-o(1)$ for $c$ with $\frac{(c-1)^2}{2c \ln 2}>1$, e.g., for $c\geq 7/2$.
The theorem follows essentially from the
Chernoff bound for random geometric variables, cf. Theorem~\ref{thm:geometricconcentration}.

\begin{proof}[Proof of Theorem~\ref{thm:wakeuprandomfailure}]
By construction, for each processor $j\in [0\ldots n-1]$ we have that path $P_j$ has at least
$\lceil \log(n-1) \rceil$ and  at most
$\lceil \log(n-1) \rceil+1$ nodes.
  Therefore the expected delay
 of path $P_j$, $\E[\REQ(P_j)]$,  equals
$(1/p) \lceil \log(n-1) \rceil$ or
$(1/p) (\lceil \log(n-1) \rceil + 1)$, respectively.
Since the variables $\{\REQ(x):x\in P_j\}$ are mutually independent,  by Theorem~\ref{thm:geometricconcentration} we have
\begin{align*}
& \Pr[\REQ(P_j) > (c/p) (\lceil \log(n-1) \rceil + 1)]\\
&\quad \leq
\Pr[\REQ(P_j) > (1+(c-1)) \E[\REQ(P_j)]]\\
&\quad \leq
\exp\left(-\frac{(c-1)^2}{2c} (\lceil \log(n-1) \rceil -1) \right)
\,.
\end{align*}
A simple union bound over all $n$ paths concludes the proof.
 \end{proof}

\subsection{Coupling the \texorpdfstring{$\GP$}{GP} and \texorpdfstring{$\WU$}{WU} Models}
\label{sec:coupling}

 To relate the time complexities of the random wakeup model and the
 $\GP$ algorithm in the presence of random node failures, we embed the failure patterns of both models in the probability space consisting of infinite binary vectors $\{\vecb \mid \vecb\in\{0,1\}^\N\}$, where the entries $\vecb_1, \vecb_2, \ldots$ are i.i.d. with a Bernoulli distribution parametrized by the success rate $p$---see, e.g., Chapter~2 of~\cite{B95book}. We need to consider infinite sequences since the number of possible failures in executions of the $\WU$ algorithm is unbounded.
The embedding of failure patterns induces distributions over executions of rumor spreading algorithms, similarly to the way randomized algorithms are presented in the classical work of Yao~\cite{Yao77}.

  In Section~\ref{subsub:patterns} we define the mappings of infinite binary vectors to failure patterns, and then to execution trees, whose heights represent time complexity of the corresponding execution of the $\GP$ and the $\WU$ protocol, respectively.
  In Section~\ref{subsec:monotonicity} we use this mapping to present our coupling argument (Lemma~\ref{lem:GP_WU}).

  Section~\ref{subsub:patterns} is quite technical. The reader only interested in the  main results may want to skip these details and jump directly to  Section~\ref{subsec:monotonicity}, considering $\HGT(\TGP(n,\vecb))$ and $\HGT(\TWU(n,\vecb))$ defined to be the time complexity of one particular execution of the $\GP$ algorithm and the $\WU$ algorithm, respectively.

\subsubsection{Failure Patterns and Execution Trees}
\label{subsub:patterns}

Any execution of the $\GP$ or of the $\WU$ algorithm with a single start processor is determined by the initial {\em system configuration} (in short {\em configuration}).
A configuration is a pair $(n,\vecb)$, where $n$ is the number of processors to which the rumor has to be delivered, and $\vecb=(b_1,b_2,\ldots)$ is an infinite binary vector representing a {\em failure pattern}.
An entry $b_i=0$ corresponds to a failed request and $b_i=1$ corresponds to a successful request. For each configuration $(n,\vecb)$, $\EXGP(n,\vecb)$ denotes the execution of the $\GP$ algorithm on $(n,\vecb)$, and $\EXWU(n,\vecb)$ denotes the execution of the $\WU$ algorithm on $(n,\vecb)$ ($\EXGP(n,\vecb)$ is always determined by the first $n$ bits of $\vecb$ - see Figure~\ref{exec_tree}, while $\EXWU(n,\vecb)$ is usually determined by a longer prefix of $\vecb$).

$\EXGP(n,\vecb)$ is defined by the
 {\em execution tree}  $\TGP(n,\vecb)$ as follows.
The vertices of $\TGP(n,\vecb)$ are configurations.
 The root of $\TGP(n,\vecb)$ is the configuration $(n,\vecb)$. If $\vecb=0 \vecc$ (for some infinite binary vector $\vecc$) then the first request sent by the execution failed. Hence in the next round there is still only one informed processor, with $n-1$ uninformed processors in its list. Thus the only child of
  $(n,\vecb)$ is $(n-1,\vecc)$. If $\vecb=1\vecc$ then the first request is successful, and hence $(n,\vecb)$ has a left child
     $(\CEIL{\tfrac{n-1}{2}},\ODD(\vecc))$ and a right child $ (\FLOOR{\tfrac{n-1}{2}},\EVEN(\vecc))$.\\
   This rule applies to all vertices of the tree: For $k>0$ and binary (infinite) vector $\vecc$, a  vertex  $(k,  0 \vecc)$ in $\TGP$  is an internal vertex with one child: $(k-1,\vecc)$, and a vertex  $(k,1  \vecc)$  has a left child   $(\CEIL{\tfrac{k-1}{2}},\ODD(\vecc))$ and a right child   $(\FLOOR{\tfrac{k-1}{2}},\EVEN(\vecc))$.
   Vertices of the form $(0,\vecc)$ are leaves.
Figure~\ref{exec_tree} illustrates the execution tree of the $\GP$ Algorithm.

 \begin{figure}[t]
\begin{center}
  \includegraphics
  [width=0.7\columnwidth]
  {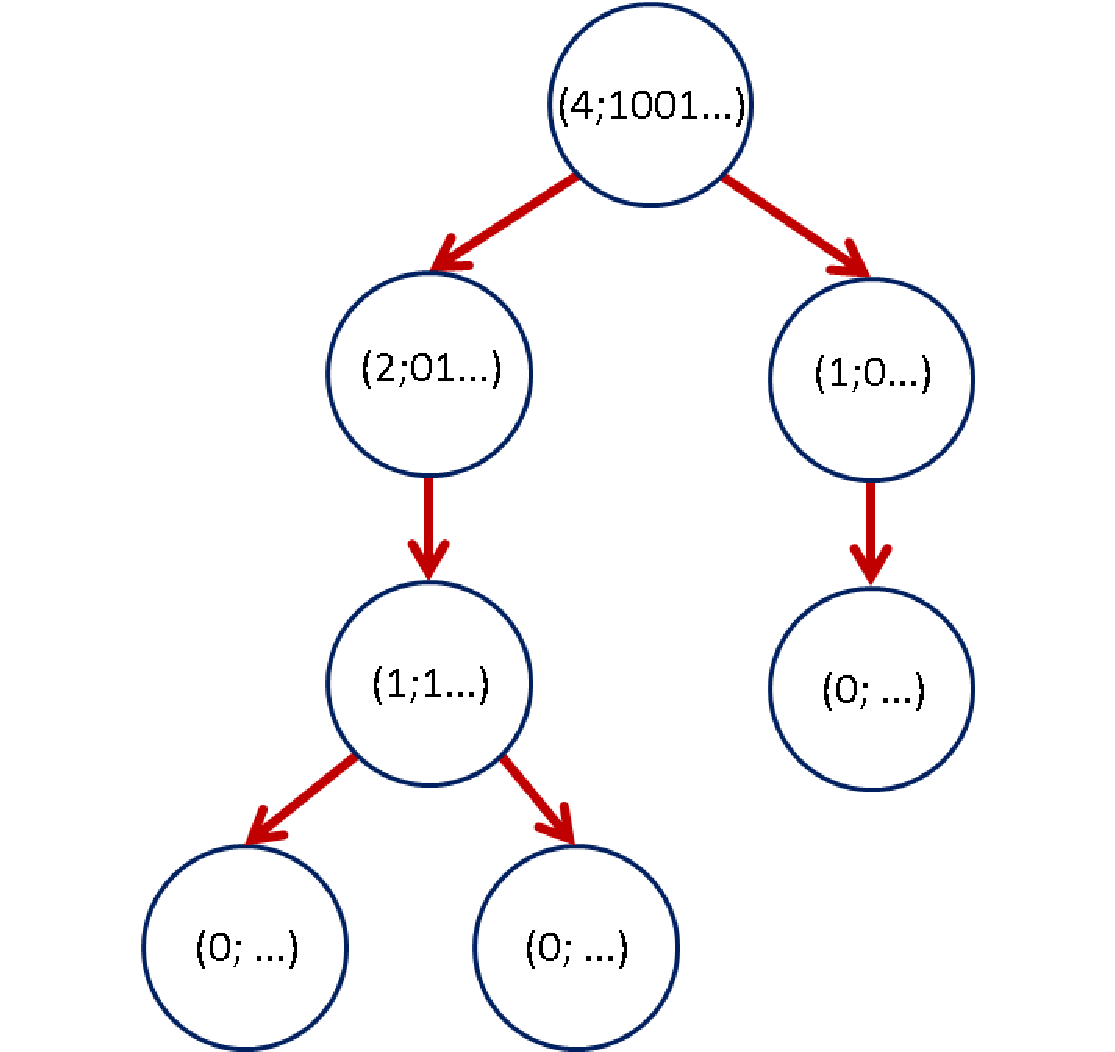}
 \caption{\textbf{$\TGP(4;1001...)$:} This tree describes the execution $\EXGP(4;1001\ldots)$ of the $\GP$ algorithm for 5 processors
 and a failure pattern $(1001...)$. Each vertex is a system configuration (${ \textit k};\vecb)$, where $k$ is the number of processors to which the rumor need to be delivered, and $\vecb$ is the corresponding failure pattern.
 }
\label{exec_tree}
 \end{center}
\end{figure}

The execution $\EXWU(n,\vecb)$ of the $\WU$ algorithm with initial configuration $(n,\vecb)$ is described by an execution graph $\TWU(n,\vecb)$ in a similar manner, with one exception:   the unique child of a vertex of the form $(k,  0  \vecc)$ for $k>0$  is  $(k,\vecc)$ (and not $(k-1,\vecc)$)---reflecting the fact that in a failed  request the number of uninformed processors remains unchanged.
Note that $\TWU(n,\vecb)$ may not be a tree, since it may contain vertices (configurations) of the form $(k,0^\N)$ which have one outgoing edge which is a self loop (corresponding to the event of infinite sequence of failed requests by a processor). It is not hard to see that $\TWU(n,\vecb)$ has no other cycles and no other directed infinite paths. Hence
$\TWU(n,\vecb)$ is a finite rooted tree or a finite rooted tree with self loops added to some of its leaves. This latter case correspond to executions in which some processor has an infinite succession of failures.

The following observation is implied by the definitions of $\TGP$ and $\TWU$.
\begin{observation}
\label{obs:coupling}
For each system configuration $(n,\vecb)$ it holds that:
\begin{enumerate}
 \item
 The  time complexity of  $\EXGP(n,\vecb)$ is $\HGT(\TGP(n,\vecb))$.
  \item
 If $\TWU(n,\vecb)$ contains a self loop, then the time complexity of  $\EXWU(n,\vecb)$ is infinite.
 The time complexity of  $\EXWU(n,\vecb)$ is $\HGT(\TWU(n,\vecb))$, otherwise.
\end{enumerate}
\end{observation}

\subsubsection{Coupling the Models}
\label{subsec:monotonicity}

Here and in the remainder of the paper we abbreviate $\HGP(n,\vecb)=\HGT(\TGP(n,\vecb))$ and $\HWU(n,\vecb)=\HGT(\TWU(n,\vecb))$.

The main coupling argument is the following lemma, whose inductive proof makes use of the fact that both functions $\HGP$ and $\HWU$ are monotone increasing in their first argument, i.e., the number of processors to be informed.

\begin{lemma}
\label{lem:GP_WU}
For each system configuration $(n,\vecb)$ it holds that $\HGP(n,\vecb)\le \HWU(n,\vecb)$.
\end{lemma}

For proving Lemma~\ref{lem:GP_WU}, we first observe that both $\HGP$ and $\HWU$ are monotone increasing in the number of uninformed processors.

\begin{lemma}
\label{lem:Monotone_n}
Let $h\in \{\HGP,\HWU\}$. The function $h$ is monotone increasing in its first argument.
\end{lemma}
This lemma follows immediately from the observation that for all $\vecb$ and $n$, $\TGP(n,\vecb)$ is isomorphic to a proper subtree of $\TGP(n+1,\vecb)$, and $\TWU(n,\vecb)$ is isomorphic to a proper subgraph of $\TWU(n+1,\vecb)$.

We are now ready to prove the main coupling argument, Lemma~\ref{lem:GP_WU}.

\begin{proof}[Proof of Lemma~\ref{lem:GP_WU}]
For $n=0$ and for all vectors $\vecb\in \{0,1\}^{\N}$ we have $$\HGP(0,\vecb)=0= \HWU(0,\vecb)\,.$$
For the all-zeros vector $\vecb=\vec{0}$ and for all $n>0$ it holds that $$\HGP(n,\vec{0})=n<\infty= \HWU(n,\vec{0})\,.$$
We proceed by induction on $n$, assuming $\vecb\neq\vec{0}$.
Let $b_k$ be the first non-zero element in $\vecb$ (for some $k\ge 1$). It follows that
\begin{align*}
 \HGP(n,\vecb)=
 & k+\max\{\HGP\left(\CEIL{\tfrac{n-k }{2}},\ODD(b_{k+1},...)\right),\\
 & \quad \quad \quad \HGP\left(\FLOOR{\tfrac{n-k }{2}},\EVEN(b_{k+1},...)\right)\}\,,
 \intertext{which, by induction hypothesis, can be bounded from above by}
  						& k+\max\{\HWU\left(\CEIL{\tfrac{n-k }{2}},\ODD(b_{k+1},...)\right),\\
  						& \quad \quad \quad \HWU\left(\FLOOR{\tfrac{n-k }{2}},\EVEN(b_{k+1},...)\right)\}\,,\\
 \intertext{which, by Lemma~\ref{lem:Monotone_n}, is itself bounded from above by}
  						& k+\max\{\HWU\left(\CEIL{\tfrac{n-1}{2}},\ODD(b_{k+1},...)\right),\\
  						& \quad \quad \quad \HWU\left(\FLOOR{\tfrac{n-1}{2}},\EVEN(b_{k+1},...)\right)\}\\
  						=&\HWU(n,\vecb)\,.
 \end{align*}
 \end{proof}

Lemma~\ref{lem:GP_WU} and Observation~\ref{obs:coupling} show that, for any initial configuration $(n,\vecb)$, the execution $\EXGP(n,\vecb)$ of the
$\GP$ algorithm
is at least as fast as the execution $\EXWU(n,\vecb)$ of the $\WU$ algorithm.
This implies that for any probability distribution $D$ on $\{0,1\}^\N$, if $\vecb$ is sampled from $D$ then $\Pr[\HGP(n,\vecb)\le H]\ge\Pr[\HWU(n,\vecb)\le H]$.
By letting $D$ be the standard distribution on $\{0,1\}^\N$ with success probability $p$, Theorem~\ref{thm:wakeuprandomfailure} easily implies the following.

\begin{theorem}
\label{thm:fnp}
Let $c>1$ be a constant.
The execution time of the $\GP$ algorithm with success probability $p \in (0,1)$ is at most
$\frac{c}{p} (\lceil \log(n-1) \rceil+1)$, with probability at least
$1-n\exp\left(-\frac{(c-1)^2}{2c} (\lceil \log(n-1) \rceil -1) \right)$.
\end{theorem}

\section{Adversarial Failures in the Randomized GP-Protocol}
\label{sec:fnm}

In this section we aim at analyzing adversarial failures.
As mentioned in Lemma~\ref{lem:GPfailure}, it has been proven in~\cite{GasPelc96} that the time complexity of the $\GP$ algorithm is at most $f+\lceil \log (n-f)\rceil$ when the number of failures is at most $f$. This bound is sharp when the first $f$ nodes fail.
For $f=\omega(\log n)$, this bound is not satisfactory in environments where random coin flips are accessible. As we shall show in this section, adding randomization allows to disseminate the rumor to all nodes in $O(\log n)$ rounds, with high probability, no matter which constant fraction of the nodes fails. Our protocol can best be described as a randomized version of the basic $\GP$ algorithm.

This randomized protocol works as follows.
When the rumor is injected at processor $0$, this processor picks a permutation $\pi \in S_{n-1}$ uniformly at random.
In round one it tries to contact processor $\pi(1)$.
If this processor has a failure, processor $0$ sends a communication request to processor $\pi(2)$ in round two.
Otherwise, i.e., if processor $\pi(1)$ is not failed,
processor $0$ sends to it the rumor and appends to this rumor the list
$\EVEN(\pi(2), \ldots, \pi(n-1))$.
Processor $0$ keeps the list $\ODD(\pi(2), \ldots, \pi(n-1))$ as its own todo-list.
The protocol continues as described in Section~\ref{subsec:GP}.
That is, all we have changed in our \emph{randomized version of the $\GP$ algorithm} is to substitute the list of processor $0$---which is
$(1,\ldots, n-1)$ in the original $\GP$ algorithm---by $(\pi(1), \ldots, \pi(n-1))$,
where $\pi$ is a random permutation of $[n-1]$.
We also have to append information on $\pi$ when transferring the rumor. It is not difficult to see (see Section~\ref{sec:fnmbits} below) that this requires a total number of $\Theta(n \log^2 n)$ bits that are appended to the rumors---compared to $\Theta(n\log n)$ in the $\GP$ algorithm. The maximum length of an individual message appendix is $n$ bits.

Here and in the remainder of this section we assume, as in all other parts of this work, that the processor initially holding the rumor, node $0$, does not fail.
Recall that in our initial node failure model, a processor either is a failed one or it does work throughout the execution.

Before we analyze the time complexity of the randomized $\GP$ algorithm, let us briefly discuss its \emph{bit complexity}; i.e., the number of bits needed to encode the lists that are appended to the initial rumor.

\subsection{The Bit Complexity of the Randomized \texorpdfstring{$\GP$}{GP} Algorithm}
\label{sec:fnmbits}

In a na\"{\i}ve implementation of the randomized $\GP$ algorithm,
every processor passes to its neighbor the
list of nodes to be informed by that processor.
As described above, in such an implementation, node 0 would pass to node $\pi(1)$ the list
$\EVEN(\pi(2), \ldots, \pi(n-1))$ of length smaller than $n/2$.
This requires $O(n \log n)$
bits to be appended to the initial rumor.
Since the length of the list halves with each successful communication request,
let the {\em level} of processor 0 be 0, and for $t>0$, a processor is at level $t$ if it was waken up by a processor at level $t-1$.
In every level of the execution tree the total number of bits that need to be communicated is $O(n \log n)$:
For $t\leq \log (n-1)$, in the $t$th level, there are
at most $2^t$ informed processors, all of which send the rumor to their descendants.
Each such processor needs to append a list of length at most $n/2^{t+1}$.
This makes a total number of $O(n \log n)$ additional bits that need to be communicated on the $t$th level.
Since there are $O(\log n)$ levels in total, the total bit complexity of this implementation is $O(n \log^2 n)$.

Another implementation of the randomized $\GP$ algorithm with the same asymptotic bit complexity but a smaller \emph{maximal} appendix is the following.
If a processor needs to communicate to its neighbor a list $L=(j_1,\ldots,j_k)$ of length $k>n/\log n$,
it appends to the rumor the {\em incidence list} of $L$; i.e. a $0/1$ vector $x$ of length $n-1$ with $x_{i}=1$ if $i \in \{j_1, \ldots, j_k\}$ and $x_i=0$ otherwise.
If a processor has received such a rumor with appended \emph{task list} $x$, it creates a random permutation $\pi_x$ of the indices $\{ i \mid x_i =1 \}$.
It then proceeds as usual, trying to spread the rumor to processor $\pi_x(1)$ in the next round.
If less than $n/\log n$ indices need to be communicated, it is cheaper to pass the list itself.
It is easily verified, using similar arguments as above, that this implementation yields a total bit complexity of $O(n \log^2 n)$.
The length of the longest appendix is linear in $n$. That is, while the total bit complexity is comparable up to a $\log n$ factor to that of the GP algorithm, the maximum message size can be significantly larger in the randomized protocol.
In Sections~\ref{sec:derandomize} and~\ref{sec:kwiseind} we describe two alternative algorithms in which the maximal size of a message appendix is in the order of $\log n$ and $\log^2 n$ bits, respectively.

\subsection{The Time Complexity of the Randomized \texorpdfstring{$\GP$}{GP} Algorithm}
\label{sec:fnmtime}

For bounding the time complexity of the randomized $\GP$ algorithm we first show that $\HGP(n,\vecb)$, the time complexity of this algorithm for given $n$ and $\vecb$, is monotone decreasing in the failure pattern $\vecb$, according to the following natural partial order on binary sequences:
 $(b_1,b_2,...) \leq (c_1,c_2,...)$ if and only if for all $i\in\mathbb{N}$ we have $b_i\leq c_i$.

\begin{lemma}
\label{lem:Monotone}
 The function $\HGP(\cdot,\cdot):\N_0 \times \{0,1\}^\N \rightarrow \R$ is monotone decreasing in its second argument. That is, for any failure pattern $\vecb$, replacing failed processors by non-faulty ones  cannot increase the time complexity.
\end{lemma}

The proof of Lemma~\ref{lem:Monotone} uses the following statement, which---informally---says that for each possible failure pattern $\vecb$, splitting the rumor spreading at the very beginning between two processors cannot increase the time complexity of the execution.

\begin{lemma}
\label{lem:splitting}
 For all $n \in \N_0$ and all $\vecb \in \{0,1\}^\N$ it holds that
\begin{align}
\HGP(n,\vecb)\ge \max\{
&\HGP(\CEIL{\tfrac{n}{2}},\ODD(\vecb)), \nonumber\\
& \HGP(\FLOOR{\tfrac{n}{2}},\EVEN(\vecb))\}\,.\label{eq:split}
\end{align}
\end{lemma}

\begin{proof}[Proof of Lemma~\ref{lem:splitting}]
The proof is
by induction on $n$. The lemma clearly holds for $n=0$ and $n=1$. So let $n\ge 2$.
Assume first that $\vecb=0 \vecc$. Then by the definition of $\TGP$,
$$\HGP(n,0\vecc) = 1+\HGP(n-1,\vecc).$$
Using the identities $\CEIL{\tfrac{k}{2}}-1=\FLOOR{\tfrac{k-1}{2}}$ and $\FLOOR{\tfrac{k}{2}}=\CEIL{\tfrac{k-1}{2}}$, we also have
\begin{align*}
&\HGP(\CEIL{\tfrac{n}{2}},\ODD(0\vecc)) = 
 1+\HGP(\FLOOR{\tfrac{n-1}{2}},\EVEN(\vecc)),
\\
&\HGP(\FLOOR{\tfrac{n}{2}},\EVEN(0\vecc)) = 
 \HGP(\CEIL{\tfrac{n-1}{2}},\ODD(\vecc))\,,
\end{align*}
which implies (\ref{eq:split}) by induction.

 The case $\vecb=1\vecc$ follows along the same lines. To simplify the notations for this case,   define $n_1=\CEIL{\tfrac{n-1}{2}}, n_2=  \FLOOR{\tfrac{n-1}{2} }, \vecd=\ODD(\vecc)$, and $\vece=\EVEN(\vecc)$. Then by the definition of $\TGP$,
 $$\HGP(n,1\vecc) = 1+\max\{\HGP(n_1 ,\vecd),\HGP(n_2 , \vece)\}.$$
And the inductive step follows from the inequalities
\begin{align*}
  & \HGP(\CEIL{\tfrac{n}{2}},\ODD(1\vecc)) \\
 & \quad = 1+
  \max\{\HGP(\CEIL{\tfrac{n_2}{2}} ,\ODD( \vece)),\\
  & \quad \quad \quad \quad \quad \quad \HGP(\FLOOR{\tfrac{n_2}{2}} ,\EVEN( \vece))\} \\
  &\quad \le   1+\HGP( n_2 , \vece) \text{ (by induction hypothesis),}
 \end{align*}
and
  $\HGP(\FLOOR{\tfrac{n}{2}},\EVEN(1\vecc)) =  
\HGP(n_1,\vecd)$.
  \end{proof}

\begin{proof}[Proof of Lemma~\ref{lem:Monotone}]
The proof  is  by induction on $n$.
For $n=0$ we have that $\HGP(0,\vecb)=0$ for all $\vecb$,  and the lemma trivially holds.
For the induction step, let $n\ge 0$ and let $\vecb,\vecc$ be two vectors such that $\vecb\le\vecc$. Then $\vecb=b_1\vecd$ and $\vecc=c_1\vece$, where $b_1\le c_1$ and $\vecd\le \vece$.
If $b_1=c_1=0$ then $\vecb=0\vecd,\vecc=0\vece$ and the induction step holds by
\begin{align*}
\HGP(n+1,0\vecb)
& =1+\HGP(n,\vecd)
\ge 1+\HGP(n,\vece)\\
& =\HGP(n+1,0\vecc)\,,
\end{align*}
where the inequality follows from the induction hypothesis. The case $b_1=c_1=1$ is similar and omitted. So we are left with the case $\vecb=0\vecd,\vecc=1\vece$ with $\vecd\le\vece$. In this case we have
\begin{align*}
 & \HGP(n+1,0\vecd) \\
  & \quad =1+\HGP(n,\vecd)\\
 &\quad \ge
 1+   \max\{\HGP(\CEIL{\tfrac{n}{2}},\ODD(\vecd)),\\
 & \quad\quad\quad\quad\quad\quad \HGP(\FLOOR{\tfrac{n}{2}},\EVEN(\vecd))\} \\
 &\quad \ge 1+ \max\{\HGP(\CEIL{\tfrac{n}{2}},\ODD(\vece)),\\
 & \quad\quad\quad\quad\quad\quad \HGP(\FLOOR{\tfrac{n}{2}},\EVEN(\vece))\}\\
 &\quad= \HGP(n+1,1\vece),
 \end{align*}
where the first inequality follows from Lemma~\ref{lem:splitting}, and the latter inequality follows from the induction hypothesis on $\ODD(\vecd),\ODD(\vece)$ and on $\EVEN(\vecd),\EVEN(\vece)$.
 \end{proof}

{\bf Note:} A similar but slightly more involved argument shows that Lemma~\ref{lem:Monotone} holds also for the function $\HWU(\cdot,\cdot)$.

\begin{theorem}
\label{thm:fnm}
Let $n \in \N\setminus \{1\}$.
Let $\varepsilon = \sqrt{\frac{\ln n}{n-1}}$.
Let $f < (n-1)(1-\varepsilon)$ and let $F \subseteq [n-1]$ of size $|F|=f$.
Let $p = 1-\frac{f}{n-1}$.
Let $c>1$ be a constant.

The probability that the randomized version of the $\GP$ algorithm has time complexity
$T \leq \frac{c}{p-\varepsilon} (\lceil \log(n-1) \rceil+1)$
is at least\\
$1-\frac{n^3}{n^2-1}\exp\left(-\frac{(c-1)^2}{2c} (\lceil \log(n-1) \rceil -1) \right)$,
even if all processors in $F$ fail.
\end{theorem}

As we mentioned after Theorem~\ref{thm:wakeuprandomfailure}, the probability bound is $1-o(1)$ for $c$ satisfying $\frac{(c-1)^2}{2c \ln(2)}>1$.
The proof of Theorem~\ref{thm:fnm} is via a reduction to the random failure model analyzed in Section~\ref{sub:wakeup}.
It makes use of several Chernoff bounds and the monotonicity proven in Lemma~\ref{lem:Monotone}. We basically show that the runtime of the randomized protocol is not worse than that of the basic deterministic $\GP$ algorithm under the presence of independent random failures. In the latter, we chose the failure probability to be slightly larger than the ``fair'' ratio $f/(n-1)$, so that, with probability at least $1-n^{-2}$, more than $f$ nodes are crashed. Combining this with the resulting runtime bound from Theorem~\ref{thm:fnp} proves Theorem~\ref{thm:fnm}.

\begin{proof}[Proof of Theorem~\ref{thm:fnm}]
It is easy to verify that the randomized $\GP$ algorithm with $f$ adversarial failures has
 the same performance as the original $\GP$ algorithm when a random subset of node failures, $R\subseteq [n-1]$ with $R=f$, is selected uniformly.
We analyze the latter.

Let $p':=p-\varepsilon$ and
$T:=\frac{c}{p'} (\lceil \log(n-1) \rceil+1)$.
Let $\vecb \in \{0,1\}^{n-1}$ with $|\vecb|_0=f$ be chosen uniformly at random.
We need to show that
\begin{align*}
& \Pr[\HGP(n,\vecb) > T]\\
& \quad \leq \frac{n^3}{n^2-1}\exp\left(-\frac{(c-1)^2}{2c} (\lceil \log(n-1) \rceil -1) \right).
\end{align*}

Let $\vecc \in \{0,1\}^{n-1}$ be such that $\Pr[\vecc_i=1]=p'$ independently for all $i \in [n-1]$. That is, the probability that $\vecc_i=0$ is $\frac{f}{n-1}+\varepsilon$, for every $i \in [n-1]$.
We show
\begin{align}
\label{cher:vecc}
\Pr[|\vecc|_0 \geq f]\geq 1-n^{-2}\,,
\end{align}
which can be easily verified by Chernoff's bound:
The expected value of $|\vecc|_0$ is $f+\varepsilon (n-1)$.
By Chernoff's bound, cf. Theorem~\ref{thm:Chernoff}, equation~(\ref{cher:i}), we have
\begin{align*}
& \Pr[|\vecc|_0 < f]\\
&\quad =
		\Pr[|\vecc|_0 < \E[|\vecc|_0]-\varepsilon (n-1)]\\
&\quad \leq
	 	\exp\left(- 2(\varepsilon (n-1))^2/(n-1) \right)\\
&\quad = 	\exp\left(- 2 \varepsilon^2 (n-1) \right)
 = 	\exp(-2 \ln n)
 = 	n^{-2}.
\end{align*}

Next we argue that
\begin{align}
\nonumber
& \Pr[\HGP(n,\vecb)\geq T]\\
& \quad \leq
\Pr[\HGP(n,\vecc) \geq T \mid |\vecc|_0 \geq f] \label{ineq:vecb:vecc}.
\end{align}
To verify (\ref{ineq:vecb:vecc}), assume $|\vecc|_0>f$.
Sample $k:=|\vecc|_0-f$ indices $i_1, \ldots, i_k$ from the $0$-positions
$\{i \in [n-1] \mid \vecc_i =0 \}$ of $\vecc$ uniformly at random.
Create $\vecd$ from $\vecc$ by replacing the zeros in positions
$i_1, \ldots, i_k$ by ones.
Then $\vecd$ is uniform in the set
$\{\vecb \in \{0,1\}^{n-1} \mid |\vecb|_0=f \}$, as is $\vecb$.
Inequality (\ref{ineq:vecb:vecc}) follows from the latter and the monotonicity of $\HGP(n,\vecb)$ in $\vecb$, as stated in Lemma~\ref{lem:Monotone}.
			
Using this inequality we bound
\begin{align*}
& \Pr[\HGP(n,\vecb) > T]
			\cdot	\Pr[|\vecc|_0 \geq f]\\
&\quad \leq 				
	\Pr[\HGP(n,\vecc) > T \mid |\vecc|_0 \geq f]
		\cdot	\Pr[|\vecc|_0 \geq f]\\
&\quad \leq 				
	\Pr[\HGP(n,\vecc) > T \mid |\vecc|_0 \geq f]
		\cdot	\Pr[|\vecc|_0 \geq f]\\
	&\quad\quad		
	+\Pr[\HGP(n,\vecc) > T \mid |\vecc|_0 < f]
		\cdot	\Pr[|\vecc|_0 < f]\\
&\quad =
	\Pr[\HGP(n,\vecc) > T]\,.						
\end{align*}
The latter quantity can be bounded by Theorem~\ref{thm:fnp}.
It shows that
the time complexity of the $\GP$ algorithm with success rate $p'$ satisfies
\begin{align*}
& \Pr[\HGP(n,\vecc) > T] \\
& \quad \leq n\exp\left(-\frac{(c-1)^2}{2c} (\lceil \log(n-1) \rceil -1) \right).
\end{align*}
Together with inequality (\ref{cher:vecc}), this concludes the proof.
 \end{proof}

\section{Reducing the Message Size}
\label{sec:derandomize}

Building on the results from the previous sections, we describe in this and the next section two alternative fault-tolerant versions of the $\GP$ algorithm that have a message overhead of only $O(\log n)$ and $O(\log^2 n)$ bits per rumor transfer, respectively. The version with logarithmic message overhead is described in this section. It requires a preprocessing phase and extra storage space at each of the processors. In Section~\ref{sec:kwiseind} we will see how the preprocessing phase can be avoided by increasing the sizes of the appended message by a logarithmic factor.

More precisely, we show here in this section that for $t \in O(n h(n)/\log n)$, $h \in \omega(1)$ being an arbitrary function tending to infinity, there are $t$~permutations such that, no matter which constant fraction of the processors fail, the probability that a permutation chosen uniformly at random out of the~$t$ yields a runtime that is greater than $c \log n$ is $o(1)$ (both the constant $c$ and the $o(1)$ failure probability will be made precise below). The algorithm is based on storing these $t$~permutations at each of the processors.

Let $\{\pi^1,\ldots,\pi^t\}\subseteq S_{n-1}$ be the stored permutations.  Upon receiving a rumor, processor 0 chooses at random an index $r\in[t]$. The algorithm now is the following minor modification of the original $\GP$ algorithm: At each round, a processor $i$ which holds a nonempty list $(j_1,\ldots,j_k)$ sends a communication request to processor $\pi^r(j_1)$, and deletes $j_1$ from its list.
If $\pi^r(j_1)$ is non-faulty, then $i$ sends it the rumor appended with
(a) the index $r$,
(b) the value $j_3$,
(c) the length $\lfloor \frac{k-1}{2} \rfloor$ of the list to be informed by processor $\pi^r(j_1)$, and
(d) the exponent $m$ of the arithmetic progression $\EVEN(j_2, \ldots, j_k)$.
Processor $\pi^r(j_1)$ starts the next round with the list $\EVEN(j_2, \ldots,j_k)$, and processor $i$ starts it with the list $\ODD(j_2, \ldots,j_k)$.

To pass information (b)--(d), $3 (\lceil \log n \rceil +1)$ bits suffice.
To pass information (a), $\lceil \log t \rceil +1$ bits are needed.
Thus, for $t \in O(n^d)$ for a constant $d$, the overall number of bits that need to be appended to the rumor is $O(\log n)$.

As mentioned above, the main goal of this section is to show (Theorem~\ref{thm:derandomized}) that for
$t \in \omega(n/ \log n)$ and suitably chosen permutations $\pi^1, \ldots, \pi^t$ this protocol, with high probability, is robust against adversarial failures.

\begin{definition}
 We call the $\GP(\pi^1, \ldots, \pi^t)$ algorithm described above \emph{$(f,r,T)$-safe}
if, for each possible failure pattern $F \subseteq [n-1]$ with $|F|=f$,
 it holds that with probability at least $r$
the runtime of the protocol $\GP(\pi^1, \ldots, \pi^t)$ with failure pattern $F$
is at most $T$.\footnote{The probability statement in this definition is with respect to the random choice of the permutation index $i \in [t]$.}
\end{definition}

Interestingly, for any constant $d<1$ and for $t \in \omega(n/\log n)$, $t$ randomly chosen permutations $\pi^1, \ldots, \pi^t$ are
$(dn, 1-o(1), O(\log n))$-safe, with high probability.

\begin{theorem}
\label{thm:derandomized}
Let $t \in \omega(n/\log n)$.
Let $\pi^1, \ldots, \pi^t$ be taken from $S_{n-1}$ independently and uniformly at random.
Let $\varepsilon := \sqrt{\ln n/(n-1)}$,
$f < (n-1)(1-\varepsilon)$, and
$p:=1-\frac{f}{n-1}$.

There are $c=c(n)\le 6+o(1)$ and $\delta=\delta(n)$ with $\lim_{n\rightarrow \infty}\delta(n)=0$, such that the probability that
$\GP(\pi^1, \ldots, \pi^t)$ is
$(f, 1-\delta , \frac{c}{p-\varepsilon} (\lceil \log (n-1) \rceil +1))$-safe is $1-o(n^{-1})$.
\end{theorem}

The proof of Theorem~\ref{thm:derandomized} is based on Theorem~\ref{thm:fnm}: By that theorem we know that, for a fixed failure set $F$ and a random permutation $\pi$, the probability that the randomized $\GP$ algorithm along permutation $\pi$ and failure set $F$ exceeds the desired runtime $T$ is less than $n^{-c}$. Based on this, we show that the fraction of $\omega(n/\log n)$ randomly chosen permutations that exceed runtime $T$ is less than $\delta$, with probability exponentially small in $n$. A union bound over all possible failure patterns $F$ concludes the proof.

\begin{proof}[Proof of Theorem~\ref{thm:derandomized}]
   By the assumption on $t$ we have a function $\delta=\delta(n)$ satisfying $\delta \cdot t> 2n/\log n$ and $\lim_{n\rightarrow \infty} \delta(n)=0$.
 Select such $\delta$ satisfying also $\delta(n)>e/n$ .

We now define $c=c(n)>1$. Fix for the moment some failure set $F \subseteq [n-1]$ of size $|F|=f$.
For given $\pi^1\ldots\pi^t$, let $T^i$ be the time complexity of the $\GP$ algorithm along permutation $\pi^i$ if all processors in $F$ fail.
Since the index $i$ is chosen uniformly, the probability that the runtime of the $\GP$ algorithm with failure set $F$ exceeds
$$T:=\frac{c}{p-\varepsilon} \left(\lceil \log (n-1) \rceil +1 \right)$$
is the fraction of indices $i \in [t]$ with $T^i > T$.
By Theorem~\ref{thm:fnm} we know that,
for a random permutation $\sigma$ of $[n-1]$, the probability that the runtime
of the $\GP$ algorithm along permutation $\sigma$ and failure set $F$ exceeds $T$ is at most
\begin{align*}
q & :=
\frac{n^3}{n^2-1}\exp\left(-\frac{(c-1)^2}{2c} (\lceil \log(n-1) \rceil -1) \right)\\
& = n\cdot\exp\left(-\frac{(c-1)^2}{2c}\log n(1+o(1))\right).
\end{align*}
We select $c=c(n)$ such that $q<n^{-2}$.
Using the fact that $\log n = \ln(n)/\ln(2)\approx 1.44 \ln(n)$, it can be easily verified that $c$ can be chosen such that $c\leq 6+o(1)$ as claimed.

We now show that for this value of $c$,
the probability that the fraction of indices $i$ with $T^i>T$ is larger than $\delta$, is exponentially small.
To obtain the statement of the theorem, we will then do a union bound over all possible choices of $F$.

The probability that for the fixed $F$
and randomly chosen permutations $\sigma^1, \ldots, \sigma^t$
at least $t_0$ of them yield a runtime exceeding
$T$
is at most
\begin{align}
	\sum_{j=t_0}^t{\binom{t}{j} q^j (1-q)^{t-j}}
 <
	\sum_{j=t_0}^t{\frac{t^j}{j!} q^j}
 <
 	\sum_{j=t_0}^t{\left (\frac{etq}{j}\right)^j},
 \label{eq:jto}	
\end{align}
where the right inequality is by Stirling approximation for $j!$.

We now set $t_0:=\CEIL{\delta t}$, which, by definition of $\delta$ is at least $2n/\log n$.
By the definition of $\delta$ and $c$, the sum above is dominated by the sum of the geometric progression $(a^{t_0},a^{t_0+1},\ldots, a^t)$ for
$a \geq eqt/t_0 \approx eq\delta$.
Since $\delta > e/n$, the value of $a$ can be chosen to be less than $1/n$, for sufficiently large $n$.
Thus, the first element in this progression, $a^{t_0}$, is smaller than $n^{-2n/\log n}=2^{-2n}$. Hence, the sum in (\ref{eq:jto}) is smaller than $2\cdot 2 ^{-2n}$.

 We now do a union bound over all possible choices of $F$.
There are at most $\binom{n}{f} < 2^{n}$ different choices.
Therefore, the probability that
for $t$ randomly chosen permutations
there exits a choice of $F$ such that the corresponding runtime is larger than
$T$, is smaller than $2^n\cdot 2\cdot 2^{-2n}=2\cdot 2^{-n}$.
 \end{proof}
Note that the definition of $\delta$ in the above proof implies that if $t(n)>2n^2/\log n$ then $\delta$ in Theorem~\ref{thm:derandomized} can be set to $\delta(n)= e/n$.

\section{Removing the Need for Preprocessing}
\label{sec:kwiseind}

In this section we show that if we are willing to compromise for messages of size $\Theta(\log^2 n)$, then the need to distribute a  set of random permutations among all processors in a preprocessing time, upon system setup, can be eliminated. That is, we provide a rumor spreading algorithm using $O(\log n)$ time and exactly $n-1$ messages, each one of them being only {$O(\log^2n)$} bits larger than the true message to be distributed, even when an adversary may crash a linear fraction of nodes before the start of the protocol. Unlike our previous algorithms, this algorithm assumes an upper bound on the number of failed processors (the algorithm will be correct also when the number of failed processors is larger, but the time analysis will not hold). At the end of this section we show that this requirement can be removed at the price of increasing the number of bits appended to each message to $\Theta(h)$, where $h=h(n)$ is any function such that $h\in\omega(\log^2 n)$.

To overcome the additional space and communication requirements of the approaches elaborated in the previous two sections we introduce a small
set of permutations on $[n-1]$ such that (i) a random permutation in this {set} is {\em sufficiently close} to a fully random permutation, and (ii) each permutation $\pi$ in this {set} has a {\em short description} $r_\pi$ of $O(\log^2 n)$  bits such that for each $j\in[n-1]$ the image $\pi(j)$ can be computed efficiently from $r_\pi$.
This latter property is the one which saves the need for appending the explicit permutation to the rumor (as opposed to the strategy in Section~\ref{sec:fnm}), or storing a set of explicit permutations in a preprocessing phase (in contrast to the protocol presented in Section~\ref{sec:derandomize}). It is instead enough to append the bits $r_\pi$ to the messages carrying the rumor.

Informally, the property of our randomized protocol that facilitates the use of such sets of permutations is that in each execution of the protocol, the time required for the rumor to reach a given processor depends with high probability on the failure pattern of only a small number ($O(\log n)$) of other processors. This feature is formalized through the notion of {\em dissemination paths} defined below.

\paragraph{Dissemination Paths:} Let $\EXGP(F,\pi)$ be an execution of the GP protocol for a given set $F$ of failed processors and a given permutation $\pi\in S_{n-1}$, and let $\TGP(F,\pi)$ be the corresponding execution tree as in Figure~\ref{exec_tree}. Each node $u$ in $\TGP(F,\pi)$ represents an event in which some processor, $\sender(u)$, sends a request to another processor, $\target(u)$. For example, let $r$ be the root node. Then for all $F$,  $\sender(r)=0$ and $\target(r)=\pi(1)$.

\begin{definition}
 Let $F\subseteq[n-1]$ and $\pi\in S_{n-1}$ be as above. Let $\TGP(F,\pi)$ be the execution tree defined by $F$ and $\pi$. For a processor $j\in[n-1]$ and a permutation $\pi$, the {\em dissemination path of $j$ given $F$ and $\pi$}, $\DIS(j,F,\pi)$, is the path in $\TGP(F,\pi)$ from the root to the unique node $v$ for which $\target(v)=\pi(j)$.
 \end{definition}
Note that in this definition we do not consider the path to processor $j$ itself but to processor $\pi(j)$. The reason for this will become evident below (see the proof of Theorem~\ref{thm:A3}).
Let us first give an informal outline of our argument:
By Theorem~\ref{thm:fnm}, for each set $F$ of failed processors of size $cn$ (for some $c<1$), and for each $j\in[n-1]$, the length of $\DIS(j,F,\pi)$ is with high probability $O(\log n)$, where the probability is taken when $\pi$ is drawn uniformly from $S_{n-1}$. Suppose now that $Q$ is a distribution over $S_{n-1}$ which looks ``close to the uniform distribution'' for an observer who may inspect at most $k$ entries of a random permutation from~$Q$ (a notion which will be made precise soon).
It is enough to show that for any failure pattern $F$, the dissemination paths $\DIS(j,F,\pi)$, $j \in [n-1]$, are still with high probability of logarithmic length, when $\pi$ is drawn from~$Q$. We show this by proving that otherwise, an observer who checks the failure status of $O(\log n)$ processors will be able to tell apart whether the permutation $\pi$ underlying the protocol is drawn from $S_{n-1}$ uniformly or from~$Q$, contradicting our indistinguishability assumption on the distribution~$Q$.
In the remaining of this section we formalize this informal argument.

The ``close to uniform" distributions which we use are called  {\it adaptive $k$-wise $\delta$-dependent distribution}.
Specifically we use constructions of such distributions that are presented in~\cite{KaplanNR09} (we could also use earlier constructions of such
distributions such as in~\cite{NaorR99}).
For a detailed exposition of these distributions
and related results see~\cite{KaplanNR09} and references within.

\paragraph{Adaptive $k$-Wise $\delta$-Dependent Distributions:}
Let $P$ and $Q$ be two distributions defined over $S_{n}$. A {\em $k$-queries adaptive distinguisher} for $(P,Q)$ is a program (which may also use randomness) $D$ which gets as input a {\em description} of a permutation $\pi\in S_{n}$, and is allowed to {\em query} at most $k$ entries of $\pi$ in an adaptive way, that is: in query $i$, for $i=1,\dots,k-1,$  it specifies $j_i\in[n]$ and receives the value of $\pi(j_i)$, and then it determines the value of $j_{i+1}$ such that $\pi(j_{i+1})$ is queried  next.
Upon termination, $D$ outputs one or zero.
We say that {\em $D$ distinguishes between $P$ and $Q$ with advantage $\delta$} if
$|\Pr_{\pi\sim Q}[D(\pi)=1 ] - \Pr_{\pi\sim P}[D(\pi)=1]| \ge  \delta$,
 where the probability is taken over the distributions $P$   and $Q$ as well as the internal randomness of the program $D$.
We say that $P$ and $Q$ are {\em adaptively ($\delta ,k$)-indistinguishable} if there is no such distinguisher $D$. Finally, we say that $Q$ is {\em adaptive $k$-wise $\delta$-dependent} if it is adaptively ($\delta,k$)-indistinguishable from the uniform distribution over $S_n$.

 The distribution $Q$ above is {\em explicit} if each permutation $\pi$ with $Q(\pi)>0$ (i.e., $\pi$ is in the support of $Q$) is indexed by a bit-string $r_\pi$
 such that the function $j\mapsto\pi(j)$ can be computed efficiently from $r_\pi$ and $j$.
 The {\it description length} of $Q$ is the maximal length (number of bits) of $r_\pi$ where $\pi$ is in the
 support of $Q$.
 The {\it time complexity} of $Q$ is the worst-case time complexity of evaluating $\pi(j)$ from $j$ and $\pi$, for all $j\in [n]$  and $\pi$ in the support of $Q$ (see~\cite{KaplanNR09}).

\begin{theorem}\cite{KaplanNR09}
\label{thm:A1}
Let $n$ be a power of two, let $\delta>0$, and let $k\in [n]$ be given.
There exists a distribution $Q_n$ over $S_n$, such that $Q_n$ is adaptive $k$-wise $\delta$-dependent,
 has description length $O(k\log(n)+\log(1/\delta))$,
and time complexity that is polynomial in $\log n$, $k$, and $\log(1/\delta)$.
\end{theorem}
\noindent
{\bf Discussion:} Theorem 5.9 of~\cite{KaplanNR09} implies Theorem~\ref{thm:A1} above for a {\em static} distinguisher, i.e., a distinguisher which has to select at the beginning of the computation the $k$-tuple $(j_1,\ldots,j_k)$ for which $\pi(j_i)$, $i=1\ldots k$, are queried (see \cite[Definition 3.3]{KaplanNR09}). The generalization to adaptive distinguishers follows by the observation 
 that a static $(\delta/n^k,k)$-indistinguishability implies an adaptive $(\delta,k)$-indistinguishability~\cite{KaplanNR09}.
To see why this implication holds, consider an adaptive distinguisher~$D$, and assume without loss of generality that $D$ makes exactly $k$ queries in each execution. Define a static distinguisher $D_s$ which start by guessing at random the $k$ queries of~$D$, each with equal probability of
$1/{n\choose k}$
and then it simulates $D$. If the guess was correct (i.e., on the given input $D$ makes exactly the same $k$ queries), $D_s$ outputs what $D$ would have. It outputs $0$ otherwise. It is not hard to see that if $D$ outputs $1$ with probability~$\varepsilon$ then $D_s$ outputs~$1$ with
probability~${\varepsilon}/{{n\choose k}}$.
So if there is no static distinguisher for $Q$ with advantage  ${\delta}/{{n\choose k}}$,
then there is no adaptive distinguisher for~$Q$ with advantage~$\delta$.

\paragraph{The Protocol:}
We first assume that $n$, the number of processors, is such that $n-1$ is a power of two. Let $\GP_{\uniform}$ be our original randomized protocol, defined in Section~\ref{sec:fnm}. Let $\delta :=1/n^2$. If we have an upper bound of $cn$, $c\in [0,1)$, for the number of node failures we wish the protocol to tolerate, we choose $k=O(\log n)$ in such a way that Theorem~\ref{thm:fnm} guarantees that the running time of $\GP_{\uniform}$ with at most $cn$ failed processors is at most $k$ with probability at least $1-\delta$.\footnote{To ease the presentation of the algorithm we assume here that an upper bound $cn$ on the number of failed processors is known. In this case, Theorem~\ref{thm:fnm} suggests to choose $k=\frac{10 \log n}{1-c-\sqrt{\ln(n)/n}}$. Estimating an a prior bound on $c$ can be avoided by setting $k=\omega(\log n)$.}

Finally, let
 $Q=Q_{n-1}$ be the adaptive $k$-wise $\delta$-dependent
 distribution guaranteed by Theorem~\ref{thm:A1}, and let $\GP_{Q}$ be the protocol identical to $\GP_{\uniform}$ except that it selects a random permutation according to distribution $Q$.

We implement $\GP_Q$ with permutation $\pi$ by the following modification of the basic $\GP$ algorithm in Section~\ref{subsec:GP}: when a processor sends a todo-list (which is an arithmetic progression represented by $O(\log n)$ bits) to another processor, it appends to it $r_\pi$, the short representation of $\pi$ (this requires additional $(O(\log^2n)$ bits). As before, if the next entry on a processor's todo-list is $j$, then the processor sends a communication request to $\pi(j)$ and deletes $j$ from its list.

In the sequel we show that for any set $F\subseteq[n-1]$ of at most $cn$  processors, the running time of $\GP_Q$ given that $F$ is the set of failed processors is {at most} $k$ {with high probability}.

\begin{theorem}
\label{thm:A3}
Let $n,c,k,\delta$ be as above and let $F$ be an arbitrary set of at most $cn$ processors.
Then the execution of the $\GP_Q$ algorithm when the processors in $F$ are faulty, terminates after at most $k$ rounds
with probability at least $1-2n\delta = 1-2/n$.
Moreover, the size of each message in the protocol is only an
additive number of $O(k \log n)$ bits larger than the original message.
\end{theorem}

\begin{proof}
First note that the ``Moreover,...'' part follows since the description length of $Q$ is $O(k\log  n )$.

For the first part of the theorem note that, by the definition of $k$ and $\delta$, we have that when running the protocol $\GP_{\uniform}$, for each $j\in [n-1]$ the length of the path $\DIS(j,F,\pi)$ is larger than $k$ with probability smaller than $\delta$.

For each processor $j\in[n-1]$, let $D_j$ be an algorithm which, upon receiving a permutation~$\pi$ simulates the execution of $\GP$ along $\DIS(j,F,\pi)$ by starting from the root of $\TGP(F,\pi)$ and following the nodes $v$ for which the todo-lists of $\sender(v)$ contain $j$.
If the the length of $\DIS(j,F,\pi)$ is less than $k$ then the algorithm outputs $0$. Otherwise, if the algorithm reaches the $k$'th node on this path, then the algorithm stops and outputs $1$.
Note that $D_j$ queries at most $k$ entries of $\pi$: For each node $v$ in $\DIS(j,F,\pi)$, $D_j$
only need to query whether $\pi(\target(v))$ is faulty, and it stops after scanning at most $k$ nodes
(here we are using the fact that $\DIS(j,F,\pi)$ is the dissemination path to processor $\pi(j)$ rather than to processor $j$).

By the definition of $k$ and $\delta$, the following holds for any $j\in[n-1]$: when $D_j$ simulates $\GP_{\uniform}$ (i.e., when $\pi$ is completely random), then it outputs 1 with probability smaller than~$\delta$.
By the fact that $Q$ is adaptive $k$-wise $\delta$-dependent, we have that for every $j\in [n-1]$, when $D_j$ simulates $\GP_{Q}$ (i.e., $\pi$ is drawn from $Q$), $D_j$ outputs 1 with probability smaller than $2 \delta$ (otherwise $D_j$ would distinguish $Q$ from the uniform distribution on $S_n$ with advantage $\ge\delta$). In other words, when $\pi$ is drawn from $Q$, it holds for every $j\in[n-1]$ that the length of $\DIS(j,F,\pi)$ is greater than $k$ with probability at most 2$\delta$.

Using the union bound for all $j\in[n-1]$ we get that the probability that there exists an $j$ such that the length of $\DIS(j,F,\pi)$ is greater than $k$ is at most $2n \delta =2/n$. This implies that the probability that $\GP_{Q}$ terminates after more than $k$ rounds is at most $2/n$.
 \end{proof}

We extend the protocol $\GP_Q$ described above to arbitrary value of $n$ (recall that the construction so far only works when $n-1$ is a power of two). For general $n$, let $n'$ be the smallest integer with $n' \ge n$ and $n'-1$ is a power of two. We let the nodes $0, \ldots, n-1$ imitate the $\GP_Q$ protocol for $n'$ nodes, however, whenever a node would call a node $i \ge n$, it does nothing and continues to the next node on its list. A run of this protocol in the presence of failed nodes $F \subseteq [n-1]$ performs exactly identical to a run of the $\GP_Q$ protocol on $n'$ nodes with failed nodes $F' = F \cup \{n, \ldots, n'-1\}$, except that the nodes $n, \ldots n'-1$ do not receive a call. Consequently, this protocol uses exactly $n-1$ messages of size $O(\log^2 n') = O(\log^2n)$ and takes $O(\log n') = O(\log n)$ rounds---only the implicit constant in the runtime is slightly larger due to the larger failure rate $|F'|/n'$, which still satisfies
$|F'|/n' = (n'-n+|F|)/n'   < 1$.

Finally, if we aim at a protocol having an $O(\log n)$ runtime with high probability for any constant fraction $cn$, $c \in [0,1)$ of failed nodes (that is, a protocol that does not depend on a known upper bound on $c$, the number of failed processors ),
we choose  $k = \omega(\log n )$.
Note that this implies that for any constant fraction of failed processors, $c$, the protocol terminates in at most $O(\log n)$ rounds with high probability. It also implies that the description length of $Q$ is $\Theta(k\log n)=\omega(\log^2 n)$.
We thus have proven the following result.

\begin{theorem}\label{thm:sec6}
Let $n \in \N$ and let $h\in\omega(\log^2 n)$.
There exists a protocol $\GP_Q$ that disseminates a rumor from an initially informed node to all other $n-1$ nodes in the network with the following properties.
(i) Exactly $n-1$ messages are sent; each message contains the rumor and an additional $\Theta(h)$ number of bits.
(ii) With high probability, all nodes are informed after a logarithmic number of rounds.
(iii) These performance guarantees remain intact if an arbitrary set of $cn$ nodes is crashed before the beginning of the process, where $c<1$ is a constant.
Alternatively, if an upper bound on the number of crashed nodes is known, the overhead per message can be reduced to $O(\log^2 n)$ bits.
\end{theorem}

\section{Conclusions and Future Work}
We have studied randomized fault-tolerant rumor spreading algorithms in complete graph topologies. The protocols can tolerate up to $cn$ initial node failures, where $c$ is an arbitrarily large constant less than one.

Our algorithms are based on introducing randomization to the elegant whispering algorithm of~\cite{GasPelc96}.
They have minimal message complexity and asymptotically optimal time complexity, do not require synchronization or activation of uninformed processors, do not need to assume an a priori bound on the number of faulty processors, and do not need an opening phase.

We proved that the time complexity of the $\GP$ algorithm in the presence of random initial failures is asymptotically optimal, i.e., $O(\log n)$.
The analysis is based on a new \emph{random wakeup model} and a novel coupling technique, which could be of independent interest.
To deal with adversarial failures, we have proposed a \emph{randomized version} of the $\GP$ algorithm.
While the randomized $\GP$ algorithm achieves best possible message complexity and asymptotically optimal time complexity, it requires up to a linear number of additional bits that need to be communicated together with the original message.
We have shown two different ways of reducing this overhead. One is to add a preprocessing step for storing $O(n h(n)/\log n)$ random permutations at the processors ($h \in \omega(1)$ being an arbitrary function tending to infinity). This decreases the message overhead to
$O(\log n)$ bits. The preprocessing phase of this protocol can be eliminated by making use of adaptive $\log n$-wise almost independent permutations,
which can be encoded with $O(\log^2 n)$ bits. An interesting problem is whether it is possible to eliminate the need for preprocessing by adding overhead of only $O(\log n)$ bits per message.

Another interesting avenue for future work is the design of whispering protocols that tolerate \emph{any-time failures} (i.e., nodes must not be crashed initially but can stop spreading the rumor at an arbitrary point in time). We are not aware of any existing work in this direction.

\subsection*{Acknowledgments.}
The authors are happy to have received very helpful comments from anonymous reviewers, in particular the suggestion to use $k$-wise $\delta$-dependent distributions for the result presented in Section~\ref{sec:kwiseind}.

We would also like to thank Ariel Gabizon for very interesting discussions and for suggesting the use of adaptive distinguishers.

Part of this work have been done while Benjamin Doerr was with the Max Planck Institute for Informatics (MPII) in Saarbr\"ucken, Germany, Carola Doerr was with the MPII and the LIAFA, Universit\'e Paris Diderot (Paris 7), France, and Shlomo Moran was a visitor at the MPII.

Carola Doerr gratefully acknowledges support from a Feodor Lynen postdoctoral research fellowship of the Alexander von Humboldt Foundation and from the Agence Nationale de la Recherche under the project ANR-09-JCJC-0067-01.

Shlomo Moran is supported by the Bernard Elkin Chair in Computer Science.
}

\newcommand{\etalchar}[1]{$^{#1}$}
\providecommand{\bysame}{\leavevmode\hbox to3em{\hrulefill}\thinspace}
\providecommand{\MR}{\relax\ifhmode\unskip\space\fi MR }
\providecommand{\MRhref}[2]{%
  \href{http://www.ams.org/mathscinet-getitem?mr=#1}{#2}
}
\providecommand{\href}[2]{#2}

\end{document}